\documentclass[11pt,final,microtype]{scrartcl}

\usepackage{algpseudocode}

\usepackage{caption}
\usepackage{subcaption}
\usepackage{amsmath}
\usepackage{amssymb}
\usepackage{amsthm}
\usepackage{subcaption}
\usepackage{circuitikz}
\usepackage{tikz}

\usepackage[top=1in, bottom=1.2in, left=1in, right=1in]{geometry}

\theoremstyle{plain}
\newtheorem{theorem}{Theorem}[section]

\newtheorem{lemma}[theorem]{Lemma}
\newtheorem{proposition}[theorem]{Proposition}
\theoremstyle{definition}

\theoremstyle{remark}

\newcommand{\ceil}[1]{\left\lceil #1 \right\rceil}
\newcommand{\ld}{\log_{2}}
\newcommand{\loq}{\log_{\varphi}}
\newcommand{\delay}{\mathrm{delay}}

\newcommand{\floor}[1]{\left\lfloor #1 \right\rfloor}
\newcommand{\sset}[1]{\left\{#1\right\}}

\definecolor{darkRed}{rgb}{0.6,0,0}
\definecolor{lightRed}{rgb}{1,0.75,0.75}
\definecolor{darkGreen}{rgb}{0,0.5,0}
\definecolor{PineGreen}{rgb}{0.01,0.5,0.45}
\definecolor{darkBlue}{rgb}{0,0,0.75}
\definecolor{CornflowerBlue}{rgb}{0.15,0,0.7}
\definecolor{lightBlue}{rgb}{0.75,0.75,1}
\definecolor{grey}{rgb}{0.5,0.5,0.5}
\definecolor{black}{rgb}{0,0,0}
\definecolor{red}{rgb}{1,0,0}
\definecolor{green}{rgb}{0,1,0}
\definecolor{blue}{rgb}{0,0,1}
\definecolor{yellow}{rgb}{1,1,0}
\definecolor{orange}{rgb}{1,0.6,0}
\definecolor{cyan}{rgb}{0,0.7,1}
\definecolor{purple}{rgb}{0.5,0,0.8}

\usetikzlibrary{shapes}
\usetikzlibrary{shapes.gates.logic.US,shapes.gates.logic.IEC,backgrounds}

\makeatletter \tikzset{circle split part fill/.style args={#1,#2}{%
    alias=tmp@name, 
    postaction={%
      insert path={ \pgfextra{%
          \pgfpointdiff{\pgfpointanchor{\pgf@node@name}{center}}%
          {\pgfpointanchor{\pgf@node@name}{east}}%
          \pgfmathsetmacro\insiderad{\pgf@x} \fill[#1]
          (\pgf@node@name.base)
          ([xshift=-\pgflinewidth]\pgf@node@name.east) arc
          (0:180:\insiderad-\pgflinewidth)--cycle; \fill[#2]
          (\pgf@node@name.base)
          ([xshift=\pgflinewidth]\pgf@node@name.west) arc
          (180:360:\insiderad-\pgflinewidth)--cycle; }}}}}
\makeatother

\title{Fast Prefix Adders\\ for Non-Uniform Input Arrival Times}
\author{Stephan Held and Sophie Theresa Spirkl \\
\normalsize {\tt \{held,spirkl\}@or.uni-bonn.de}\\
  Research Institute for Discrete Mathematics, University of Bonn}
\date{}
\begin{document}

\maketitle

\begin{abstract}
\noindent {\bf  Abstract}

  We consider the problem of constructing fast and small parallel
  prefix adders for non-uniform input arrival times. This problem
  arises whenever the adder is embedded into a more complex circuit,
  e.\ g.\ a multiplier.
  
  Most previous results are based on representing binary
  carry-propagate adders as so-called parallel prefix graphs, in which pairs of
  generate and propagate signals are combined using complex gates
  known as prefix gates. Adders constructed in this model usually
  minimize the delay in terms of these prefix gates. However, the delay
  in terms of logic gates can be worse by a factor of two.
  
  In contrast, we aim to minimize the delay of the underlying logic
  circuit directly.  We prove a lower bound on the delay of a carry
  bit computation achievable by any prefix carry bit circuit and
  develop an algorithm that computes a prefix carry bit circuit with
  optimum delay up to a small additive constant. Furthermore, we use
  this algorithm to construct a small parallel prefix adder.

  Compared to existing algorithms we simultaneously improve the delay
  and size guarantee, as well as the running time for constructing
  prefix carry bit and adder circuits.
\end{abstract}

\section{Introduction }

The addition of binary numbers is one of the most fundamental
computational tasks performed by computer chips.
Given two binary addends $A = (a_n \dots a_1)$ and $B = (b_n \dots b_1)$,
where index $n$ denotes the most significant bit, their sum $S = A+B$
has $n+1$ bits. For each position $1 \leq i \leq n$, we compute a
\emph{generate signal} $g_i$ and a \emph{propagate signal}
$p_i$, which are defined as follows:
\begin{equation}
\begin{array}{rl}
g_i &= a_i \wedge b_i,\\
p_i &= a_i \oplus b_i,
\end{array}
\label{eqn:generate-and-propagat}
\end{equation}
where $\wedge$ and $\oplus$ denote the binary \textsc{And} and
\textsc{Xor} functions.  The \emph{carry bit} at position $i+1$ can be
computed recursively as $c_{i+1} = g_i \vee (p_i \wedge c_i)$ \cite{Kno01,WS58}.
From the carry bits, we can compute the output $S$ via $s_i = c_i \oplus p_i$ for $1 \leq i\leq n$ and $s_{n+1} = c_{n+1}$.

For two pairs $(g_i, p_i)$ and $(g_j, p_j)$ of generate and propagate
signals, we define a binary \emph{prefix operator} as
\begin{equation}
{g_i \choose p_i} \circ {g_j \choose p_j} = {g_i \vee (p_i \wedge g_j)
  \choose p_i \wedge p_j}.
\label{eqn:prefix-operator}
\end{equation}

This operator is associative, and it can be used to compute the
carry bit $c_{i+1}$ using the identity
$${c_{i+1} \choose p_i \wedge p_{i-1} \wedge \dots \wedge p_1} = {g_i
  \choose p_i} \circ {g_{i-1} \choose p_{i-1}} \circ \dots \circ {g_1
  \choose p_1}.$$ 

The prefix operator allows us to simplify notation by combining
generate and propagate signals into a single term $z_i= (g_i,p_i)$ and
computing $c_{i+1}$ as the first component of $z_i \circ \dots \circ
z_1$.  Figure~\ref{fig:pfx-to-gate3} shows a prefix gate computing $z
\circ z'$ for the prefix operator in (\ref{eqn:prefix-operator}) on
the left and its underlying logic circuit on the right.

  \begin{figure}[!tb]%
    \centering{%
      \resizebox{0.4\linewidth}{!}{%
        \begin{tikzpicture}
          \input{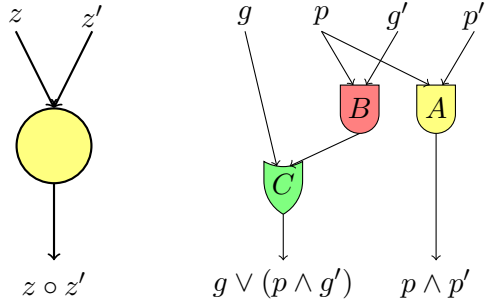}
        \end{tikzpicture}%
      }%
      \caption{Prefix gate and underlying logic circuit}%
      \label{fig:pfx-to-gate3}%
    }%
  \end{figure}%

Formally, a \emph{logic circuit} is a non-empty connected acyclic
directed graph consisting of nodes that are either \emph{inputs} with
at least one outgoing edge and no incoming edges, \emph{outputs} with
exactly one incoming edge and no outgoing edges, or \emph{gates} with
one or two incoming edges representing one of the 2-bit logical
functions \textsc{And} ($\wedge$), \textsc{Or} ($\vee$), \textsc{Xor}
($\oplus$), \textsc{Not} and their negations.%

The number of gates is the \emph{size} of the
circuit. The \emph{(maximum) fan-out of the circuit} is the maximum
fan-out (out-degree) of its nodes. The \emph{depth of the
  circuit} is the maximum number of gates on a directed path.

A logic circuit with inputs $g_1,p_1,\dots,g_n,p_n$ is called 
a \emph{prefix carry bit circuit}  if it  computes  $c_{n+1}$ and $p_1\wedge\dots\wedge p_n$, 
it is built from prefix operator gadgets in Figure~\ref{fig:pfx-to-gate3}, and the subcircuit computing $p_1\wedge\dots\wedge p_n$ is a tree.
Similarly, a \emph{prefix adder}  is a logic circuit built using the gadgets in Figure~\ref{fig:pfx-to-gate3}
that computes  $c_{i+1}$  and $p_1\wedge\dots\wedge p_i$ for all $i=1,\dots,n$ at its $2n$ outputs.

A graph that arises from a prefix carry bit circuit by contracting
each gadget into a \emph{prefix gate} as in
Figure~\ref{fig:pfx-to-gate3}, and by contracting all input pairs
$(g_i,p_i)$ into $z_i$ and the output pair
$(c_{n+1},p_n\wedge\dots\wedge p_1)$, is called a \emph{prefix tree}.
Likewise, a \emph{parallel prefix graph} arises from a prefix adder by
contracting each gadget, all input pairs $(g_i,p_i)=z_i$ and output
pairs $(c_{i+1},p_1\wedge\dots\wedge p_i)$ for all $i = 1,\dots,n$.
For inputs $z_1, \dots, z_n$, a prefix tree computes the last
carry bit of an addition $z_n \circ \dots \circ z_1$, while a
\emph{parallel prefix graph} computes $z_i \circ \dots \circ z_1$ for
all $1 \leq i \leq n$, i.\ e.\ all carry bits of an addition.

When aiming for a bounded fan-out, we allow the use of repeater gates
with fan-in one (a single incoming edge) and fan-out at least two in
all types of circuits and graphs.

An example of the transition between parallel prefix graphs and prefix
adders is given in Figure~\ref{fig:transition}. On the left the serial
parallel prefix graph with depth 3 is replaced by an
\textsc{And-Or}-path with logic circuit depth 6 known as the
ripple-carry adder. For the Kogge-Stone parallel prefix graph
\cite{kogge-stone} on the right, the depth increases from two to four
and the maximum fan-out increases from two to three.
\begin{figure}[tbh]
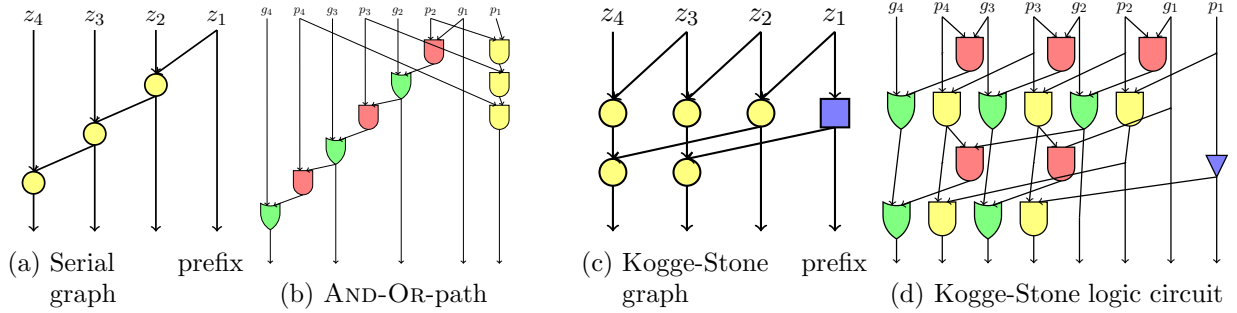

  \centering
  \begin{subfigure}[b]{.19\linewidth}
    
  \centering{
    \resizebox{1\linewidth}{!}{
      \begin{tikzpicture}
        \input{serial-small}
      \end{tikzpicture}
    }
    \caption{Serial prefix graph}
    \label{fig:serial-small}%
  }

  \end{subfigure}
  \begin{subfigure}[b]{.21\linewidth}
    
  \centering{
    \resizebox{1\linewidth}{!}{
      \begin{tikzpicture}
        \input{and-or-path}
      \end{tikzpicture}
    }
    \caption{\textsc{And-Or}-path}
    \label{fig:and-or-path}%
  }

  \end{subfigure}
  \quad\quad
  \begin{subfigure}[b]{.23\linewidth}
    
  \centering{
    \resizebox{1\linewidth}{!}{
      \begin{tikzpicture}
        \input{kogge-small}
      \end{tikzpicture}
    }
    \caption{Kogge-Stone prefix graph}
    \label{fig:kogge-small}%
  }

  \end{subfigure}
  \begin{subfigure}[b]{.29\linewidth}
    
  \centering{
    \resizebox{1\linewidth}{!}{
      \begin{tikzpicture}
        \input{kogge-logic}
      \end{tikzpicture}
    }
    \caption{Kogge-Stone logic circuit}
    \label{fig:kogge-logic}%
  }

  \end{subfigure}
  \caption{Prefix graphs as logic circuits}
  \label{fig:transition} 
  \label{fig:kogge-transition} 
\end{figure}

Additions are typically not performed as isolated
tasks, but the input signals result from preceding computational
stages and become available at different fixed \emph{arrival times}
$t_i\in \mathbb{N}_0$ $(i\in \{1,\dots,n\}$), e.\ g.\ when used within a
multiplier.  Here we make the simplifying assumption that $g_i$ and
$p_i$ have the same arrival time at the inputs, which is essentially
fulfilled if they are generated as in
(\ref{eqn:generate-and-propagat}).  We define the delay of a directed
path in a  logic circuit starting at an input as its depth plus its input
arrival time.  The delay of a vertex is the maximum delay of a path
ending in the vertex and the \emph{delay of the circuit} is the
maximum delay of its outputs. Depth and delay coincide if all input arrival times are zero.
Paths and outputs attaining the delay of
the circuit are called \emph{critical}.  The delay of all vertices can
be computed in linear time by a longest path computation in an acyclic
network. 

\begin{figure}[tbh]
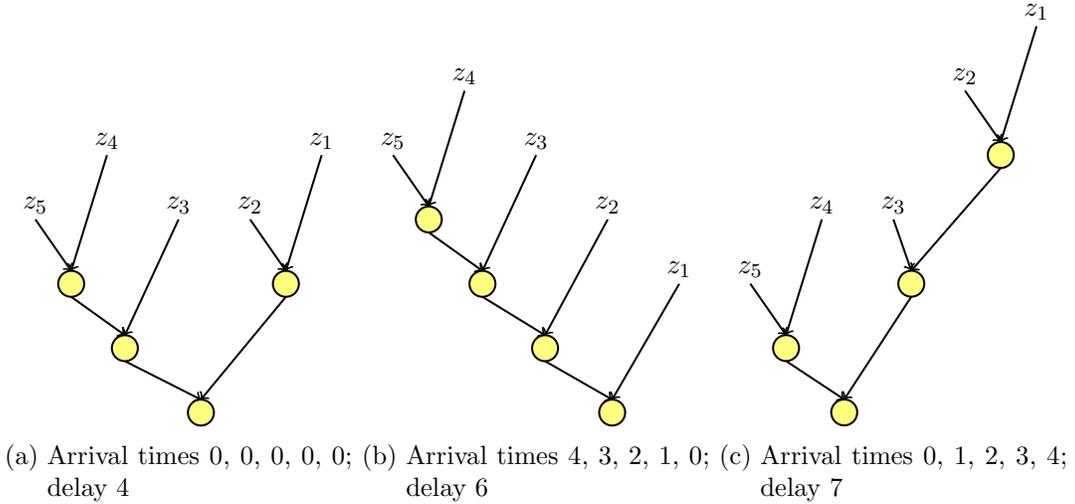

  \centering
  \begin{subfigure}[b]{.28\linewidth}
    
  \centering{
    \resizebox{1\linewidth}{!}{
      \begin{tikzpicture}
        \input{example1}
      \end{tikzpicture}
    }
    \caption{Arrival times 0, 0, 0, 0, 0; delay 4}
    \label{fig:example1}%
  }

  \end{subfigure}
  \begin{subfigure}[b]{.28\linewidth}
    
  \centering{
    \resizebox{1\linewidth}{!}{
      \begin{tikzpicture}
        \input{example2}
      \end{tikzpicture}
    }
    \caption{Arrival times 4, 3, 2, 1, 0; delay 6}
    \label{fig:example2}%
  }

  \end{subfigure}
  \begin{subfigure}[b]{.28\linewidth}
    
  \centering{
    \resizebox{1\linewidth}{!}{
      \begin{tikzpicture}
        \input{example3}
      \end{tikzpicture}
    }
    \caption{Arrival times 0, 1, 2, 3, 4; delay 7}
    \label{fig:example3}%
  }

  \end{subfigure}
  \caption{Different arrival times profiles and their optimum prefix trees}
  \label{fig:at-ex}
\end{figure}

In Figure~\ref{fig:at-ex}, we show an example with five inputs and its
optimum solutions for different arrival time patterns. Each tree is
optimal for neither of the other two arrival times sequences.

We aim for a prefix carry bit circuits and adders with close to minimum delay
and small size. 

A minimum-depth prefix graph for uniform input arrival times is given
by \cite{kogge-stone, ladner} and has depth $\lceil\ld n\rceil$ in
terms of prefix gates, but a non-minimal depth of $2\lceil{\ld
  n}\rceil$ as a logic circuit. For non-uniform arrival times, these
circuits might be by a factor of three worse than the lower bound, for
example for the arrival time pattern $t_1 = \ld n$ and $t_2 = \dots =
t_n = 0$. In Figure~\ref{fig:kogge-small}, if $z_1$ has arrival time
two and all other arrival times are zero, the delay of
Figure~\ref{fig:kogge-logic} is $6$.

Parallel prefix graphs minimizing the overall prefix graph delay for
special input arrival time patterns that occur mostly in certain
multipliers were presented in \cite{zimmermann,okl}.

An algorithm for constructing optimum-delay parallel prefix graphs for
arbitrary non-uniform input arrival times is given in \cite{choi},
however this approach may require $\mathcal{O}(n^2)$ gates for a full
$n$-bit adder.  In \cite{roy} parallel prefix adders are enumerated
with heuristic pruning to achieve good performance-area tradeoffs in
practice.  All these approaches minimize the delay of the prefix graph
rather than the underlying logic circuit.  As the prefix operator
contains two subsequent gates, the resulting delay of the underlying
circuit may be worse by a factor of two.

As it is common practice in logic synthesis
\cite{kogge-stone,ladner,bonnlogic,roy}, we use a simple technology-independent circuit and delay model in this work.  In
hardware, the delay of a gate certainly depends on its physical
structure.
In CMOS technology, \textsc{Nand}/\textsc{Nor} gates are faster than
\textsc{And}/\textsc{Or} gates and efficient implementations exist for
integrated multi-input {\sc And-Or}-Inversion gates and {\sc Or-And}-Inversion
gates. 
We assume that circuits are re-mapped into logically equivalent circuits
based on technology specific delays in a {\it technology mapping} step
\cite{keutzer88,Chatterjee+techmap2006} after logic synthesis.
Despite its simplicity, the simple circuit and delay model is
successfully used in practice for re-optimizing carry bit functions
even late in the design flow \cite{bonnlogic}.

\subsection{Our contribution}
We will use the delay properties of the prefix operator
(\ref{eqn:prefix-operator}) aiming to minimize the delay of logic
circuits for additions instead of the corresponding prefix
graphs. This idea was used by \cite{bonn1}, who proposed a cubic-time
dynamic programming algorithm to compute a fast carry bit circuit.

With a deeper structural analysis of near-optimum prefix trees in
Section~\ref{sec:bl-core}, we can construct a carry bit circuit with a
better delay bound, size, and running time as shown in the
rows with type ``Carry'' of Table~\ref{tbl:bl-improve}.

In Section~\ref{sec:bl-adder}, we apply the carry bit algorithm to
substantially improve the delay bound given in \cite{bonn2} for a full
$n$-bit adder with input arrival times $t_1, \dots, t_n \in
\mathbb{N}_0$. The result is listed in the rows with type ``Adder'' in
Table~\ref{tbl:bl-improve}.

\begin{table}[!b]
\small
  \begin{center}
  {
      \begin{tabular}{|c|c|c|c|c|c|}
        \hline
                     & \textbf{Type}  & \textbf{Delay} & \textbf{Size} & \textbf{Fan-out} & \textbf{Running Time} \\
        \hline
        \cite{bonn1}       & Carry  & $1.441W+3$ & $4n-3$ & $1,2,3$  &  $\mathcal{O}(n^3)$ / $\mathcal{O}(n^3\log n)$ \\
        \textbf{Here} & Carry & $1.441W+2.674$ & $3n-3$ & $1,2$ &  $\mathcal{O}(n \log n)$ / $\mathcal{O}(n \log^2 n)$ \\
        \hline
        \cite{bonn2}       & Adder  & $2W + 6 \ld \ld n +
        \mathcal{O}(1)$ & $6n\ld \ld n$ & $\sqrt{n} + 1$  & $\mathcal{O}(n^2)$ / $\mathcal{O}(n^2\log n)$\\
        \textbf{Here} & Adder  & $1.441W+5 \ld \ld n + 4.5$
        & $6n\ld \ld n$ & $\sqrt{n} + 1$  & $\mathcal{O}(n \log n)$ / $\mathcal{O}(n \log^2 n)$ \\
        \hline
      \end{tabular}
}
\end{center}
\caption{Improvements over \cite{bonn1,bonn2}, where $W = \ld \left(\sum_{i=1}^n 2^{t_i}\right)$ is a lower bound for the delay. Running times 
assume  constant/linear time for binary addition.}
\label{tbl:bl-improve}
\end{table}

Finally, in Section~\ref{sec:lb}, we prove a lower bound on the delay
of any prefix carry bit circuit, which shows that our carry bit
algorithm is delay-optimal up to an additive constant of 5.

\section{Algorithm for Single Carry Bit Circuits} \label{sec:bl-core} 

We start with a method that given a parallel prefix graph allows us to
compute the delay of the underlying logic circuit up to an
additive error of one.

\begin{proposition} \label{lem:pfx-to-logic} Given a parallel prefix graph
  or prefix tree, we propagate the arrival times (which might all be
  zero) through the prefix gates so that the delay $t$ of a gate with
  left input (higher indices) $l$ and right input (lower indices) $r$
  with delay $t_l$ and $t_r$, respectively, is defined as $t =
  \max\{t_r+2, t_l+1\}.$ Let $d$ be the maximum delay computed with
  this procedure, maximized over all gates, inputs and outputs, then
  the delay $D$ of the logic circuit corresponding to the given prefix
  graph or prefix tree satisfies
  $d \leq D \leq d+1.$
\end{proposition}
\begin{proof}
  This is a  consequence of a longest path computation in acyclic networks.
  We construct a logic circuit from the prefix graph. For every input
  pair $(g, p)$ corresponding to an input with arrival time $t$, we
  set the arrival time of $g$ and $p$ to be $t$ and $t-1$,
  respectively. We prove by induction that for every signal pair $(g,
  p)$ corresponding to a signal $z$ computed in the prefix graph, $g$
  has delay at least one more than $p$, and if $z$ is not an input,
  the delay of $g$ is the maximum of two plus the delay of the
  generate signal of its right predecessor and one plus the delay of
  the generate signal of its left predecessor. This is clear for
  inputs. Now consider a signal $z \circ z'$ as in
  Figure~\ref{fig:pfx-to-gate3}. Let $t_{g}, t_{p}, t_{g'},$ and
  $t_{p'}$ denote the delay of $g, p, g'$ and $p'$,
  respectively. By induction hypothesis, $t_{p'} + 1 \leq t_{g'}$
  and $t_{p} + 1 \leq t_{g}$. Therefore, $g \vee (p \wedge
  g')$ has delay $\max \sset{t_{p} + 2, t_{g'} + 2, t_{g} + 1} =
  \max \sset{t_{g'} + 2, t_{g} + 1}$. Furthermore, $\max
  \sset{t_{p} + 2, t_{g'} + 2, t_{g} + 1} \geq 1 + \max
  \sset{t_{p} + 1, t_{p'} + 1}$, which proves that $p' \wedge p$
  is indeed by at least one time unit earlier. 

  Inductive application of the argument above yields that the generate
  signal of every output arrives at time $\leq d$ under the assumption
  that all propagate signals of the inputs arrive one time unit
  earlier than their actual arrival time. Shifting all computed delays
  up by one time unit yields that the delay of the logic circuit for
  the actual arrival times is at most $d+1$.

  To show that $d \leq D$, consider only the generate signals, i.\ e.\
  using the notation of Figure~\ref{fig:pfx-to-gate3}, consider a
  logic circuit $G$ in which all gates of type $A$ and inputs $p_i$
  computing propagate signals are removed, and gates of type $B$ are
  replaced by repeaters. It follows that for every output $c$, the
  subcircuit of $G$ which consists of all ancestors of $c$ is a
  tree. This is certainly true in the parallel prefix graph, and after
  the removal of propagate signals, every prefix gate internally
  corresponds to a tree as well. Removing gates and inputs from a
  circuit does not increase its delay, because a critical path in $G$
  is also a path in the original circuit.

  Computing the delay of a signal in a tree is easy: when combining
  the generate signals of two inputs, one of them has to pass through
  two gates (a repeater instead of $B$, and $C$ and the other has to
  pass through only one (namely $C$). This shows that the given method
  for computing $d$ indeed yields a lower bound.
\end{proof}
For uniform arrival times, $d = D$, but for arbitrary arrival times,
$D = d+1$ is possible, for example by choosing the arrival times of
$z$ and $z'$ in Figure~\ref{fig:pfx-to-gate3} as 1 and 0, respectively. 

The prefix graph and its underlying logic circuit can
vary greatly in depth and delay. For example, a prefix graph of
optimal depth $\ceil{\ld n}$ as in Figure~\ref{fig:kogge-small}
contains a balanced binary tree computing its last output, therefore
its logic circuit depth is $2\ceil{\ld n}$. However, the depth only
doubles for the lower (right) input of a prefix gate by
Lemma~\ref{lem:pfx-to-logic}, which we exploit in the following.

For a single carry bit computation with arrival times, Rautenbach et
al.\ \cite{bonn1} give a dynamic programming algorithm with cubic
running time.  The algorithm restructures an \textsc{And-Or}-path
similar to a prefix tree. Here the right-to-left ordering of the
leaves of this tree is fixed as $z_1, \dots, z_n$, because $\circ$ is
not commutative. The algorithm recursively splits the sequence of
inputs into two parts at an index $l$ attaining the minimum in the
recursive delay function
\begin{equation}
\mathcal{D}(t_1, \dots, t_n) = \min_{l=1, \dots,
  n-1} \max \left\{\mathcal{D}(t_1, \dots, t_l) +2,
  \mathcal{D}(t_{l+1}, \dots, t_n) + 1\right\}.
\label{eqn:bl-dp}
\end{equation}
This solution can be computed for every subsequence $t_i, t_{i+1}, \dots, t_j$ of indices
via dynamic programming by choosing the $\mathcal{D}$-optimum position
$l$ at which to split the sequence, which yields the following
result. 
\begin{theorem}[{\cite{bonn1}}] \label{thm:bl-main} For $n$ input pairs $(g_i, p_i)$ for $1 \leq i \leq n$
  with arrival times $t_1, \dots, t_n \geq 0$, there is a logic circuit computing the carry bit $c_{n+1}$ with
\begin{equation} \delay(c_{n+1}) \leq 1.441
  \ld \left(\sum_{i=1}^n 2^{t_i}\right) + 3.
  \label{eqn:bl-delay-bound}
\end{equation} This circuit can be
  constructed in $\mathcal{O}(n^3)$ time. It has size at most $4n -
  3$, and its maximum fan-out is bounded by two at all gates and
  bounded by three at all inputs.
\end{theorem}
Using our definition of a prefix tree, the size of the carry bit circuit can be reduced by $n$.
\begin{lemma} \label{lem:small-bl} Any prefix tree computing a single
  carry bit has an underlying  logic circuit  size of at most $3n-3$ and an underlying 
  maximum fan-out of two.
\end{lemma} 
\begin{proof}
  This is clear as any   prefix tree for $n$ inputs  has exactly $n-1$ prefix gates. 
\end{proof}
To analyze the structure of fast prefix carry bit circuits we
begin with a well-known definition: let $F_n$ be the $n$-th
\emph{Fibonacci number}, where $F_0 = 0, F_1 = 1$ and $F_n = F_{n-1} +
F_{n-2}$. The exact formula for computing the $n$-th Fibonacci number
is $F_n = \frac{1}{\sqrt{5}} (\varphi^n - \psi^n)$, where $\varphi =
\frac{1+ \sqrt{5}}{2}$ is the golden section and $\psi =
\frac{1-\sqrt{5}}{2}$.

We first prove a similar delay bound to \cite{bonn1}, but instead of
bounding the recursive function $\mathcal{D}$, we explicitly construct
our solution and obtain useful structural information about it.

\begin{lemma} \label{lem:fib-pfx} Let $t_1, \dots, t_n \in
  \mathbb{N}_0$ be a sequence of input arrival times for inputs $z_1,
  \dots, z_n$, and let $F_k$ be the first Fibonacci number that is
  at least as large as $\sum_{i=1}^{n} (F_{t_i + 3} - 1)$. Then there
  is a prefix tree computing $z_n \circ \dots \circ z_1$ with
  logic gate delay at most $k$.
\end{lemma} 

\begin{proof}
Throughout the proof, we implicitly assume that every propagate signal
actually arrives by at least one time unit earlier than the
corresponding generate signal, i.\ e.\ $g_i$ has arrival time $t_i$
and $p_i$ has arrival time at most $t_i - 1$, thus prefix gates have
depth two for the input with smaller indices and depth one for the
input with larger indices. Under this assumption, we prove that the
delay is at most $k-1$. Adding one to all arrival times and delays
yields a circuit with delay $k$ under the assumption that $g_i$ is
available at time $t_i + 1$ and $p_i$ at time $t_i$, which is true for
the given input arrival times. For signal pairs $(g_i,p_i)$ with
arrival times satisfying this assumption, we say that they have
\emph{skewed arrival times}.\index{arrival time!skewed}

  The proof has two main parts. In the first part, we construct a
  binary tree $T$ with $F_k$ leaves in such a way that if we consider
  its internal nodes as prefix gates and its leaves as inputs with
  arrival time 0, then its overall delay is $k-1$. During the second
  step, we replace sections of consecutive leaves and the
  corresponding subtrees of $T$ with our original inputs so that the
  arrival time of the input does not exceed the depth of the subtree.

  Let $T$ be a tree constructed by starting at the root $r$ and
  recursively constructing a binary tree with $F_{k-1}$ leaves on the
  left and one with $F_{k-2}$ leaves on the right as in
  Figure~\ref{fig:fib-tree}. We refer to $T$ as a \emph{Fibonacci
    tree} for $k$.

  \begin{figure}[!tb]%
    \centering{%
      \resizebox{0.55\linewidth}{!}{%
        \begin{tikzpicture}
          \input{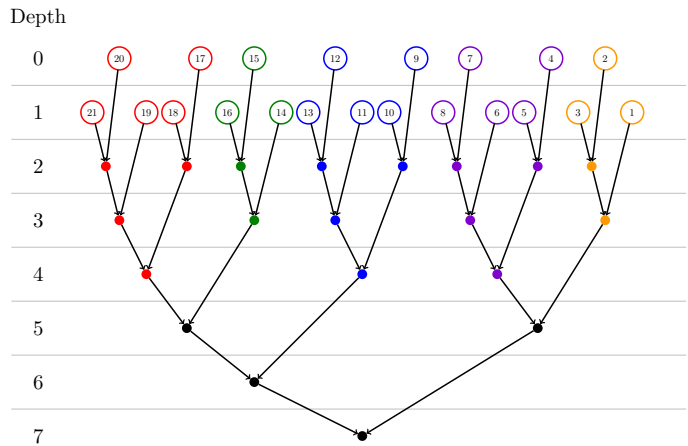}
        \end{tikzpicture}%
      }%
      \caption{Fibonacci tree $T$ for $k = 8$}%
      \label{fig:fib-tree}%
    }%
  \end{figure}%

  Replacing all non-leaf nodes with prefix gates and leaves with new
  inputs (with arrival time 0 and unrelated to the original inputs) as
  well as adding an output at the root yields a prefix tree for $F_k$
  inputs with logic gate depth $k-1$. This can be seen inductively; it
  is certainly true for $k = 2, 3$ and thus for $k > 3$, the left tree
  has depth $k-2$, the right tree has depth $k-3$, and the last prefix
  gate has delay $\max\{k-2+1, k-3+2\} = k-1$. The minimum depth of a
  prefix tree with $l$ leaves is at most $k-1$ if and only if $l \leq
  F_k$.

  Now we show how to replace parts of the tree by inputs with skewed
  arrival times $t_1, \dots, t_n$ without increasing the delay. We
  start by subdividing the leaves of the tree: from right to left, the
  first $F_{t_1+3}-1$ leaves are assigned to the first input, the next
  $F_{t_2+3}-1$ leaves are assigned to the second input, and input $i$
  gets leaves $1+\sum_{j=1}^{i-1} (F_{t_j+3} -1)$ up to $\sum_{j=1}^i
  (F_{t_j +3} -1)$. Our choice of $k$ ensures that every input $i$
  gets $F_{t_i +3}-1$ successive leaves assigned to it; leftover
  leaves can be deleted without increasing the delay. The ordering of
  the inputs is preserved within the tree.

  We define a subtree of size $l$ to be a tree obtained by taking a
  vertex $v$ and all its successors with $l$ leaves in total. By
  construction, every subtree of size $l$ must be a Fibonacci tree for
  some $j$ with $F_j = l$. Furthermore, for every $F_j$ with $j \leq
  k-1$, we can find subtrees of $T$ of size $F_j$. A vertex $v$ in $T$
  is the root of a subtree of size $F_j \neq 1$ if and only if $v$ has
  depth $j-1$. For $F_j = 1$, we know that $j \in \sset{0,1}$ and $v$
  has depth $0$.

  Our goal is to show that every input $i$ with arrival time $t_i$
  owns all the leaves of a subtree of size $F_{t_i+1}$. In order to
  see this, we remove all edges connecting a vertex with depth at most
  $t_i$ to a vertex with depth more than $t_i$ from the tree. This
  separates the tree into a connected component containing the root
  and several subtrees of size at most $F_{t_i+1}$. For example, if
  $t_i = 4$, then Figure~\ref{fig:fib-tree} would contain the
  component containing the root as well as subtrees indicated by the
  coloring of size $3,5,5,3,5$ in that order. In general, since every
  gate has depth 1 or 2, each root of such a tree has depth $t_i$ or
  $t_{i-1}$, therefore the subtrees can only have size $F_{t_i + 1}$
  or $F_{t_i}$. Our next goal is to prove that this ordered
  \emph{subtree sequence} has a special structure. Since only the
  roots of ``big'' subtrees of size $F_{t_i+1}$ can be replaced by
  input $i$ without increasing the delay, we show that there are few
  small subtrees of size $F_{t_i}$.

  Due to the fact that the depth difference between a node and its
  left child is always one, the leftmost root in the subtree sequence
  of a Fibonacci tree for some $k \geq t_i$ has depth $t_i$ and its
  parent has depth $t_{i+1}$. Therefore, the subtree rooted here has
  size $F_{t_i+1}$. We will now show that in a Fibonacci tree, the
  ordered subtree sequence of the trees of size $F_{t_i+1}$ and size
  $F_{t_i}$ never contains two consecutive subtrees of size
  $F_{t_i}$. For $k = t_i + 1$, this is clear. For $k = t_i + 2$,
  there are only two subtrees, and the left one has size
  $F_{t_i+1}$. For $k>t_i + 2$, the subtree sequence of a Fibonacci
  tree for $k$ corresponds to the concatenation of the subtree
  sequences corresponding to a tree for $k-1$ and a tree for $k-2$. As
  those satisfy the claim by induction hypothesis and each sequence
  starts with a tree of size $F_{t_i + 1}$, the Fibonacci tree for $k$
  has the stated property as well.

  We know that input $i$ owns $F_{t_i+3}-1$ consecutive leaves. In the
  subtree sequence, at most the first $F_{t_i+1}-1$ leaves belonging to
  input $i$ are part of subtrees of which $i$ does not own the first
  (rightmost) leaf. Of the remaining leaves, the first $F_{t_i}$
  might cover a subtree of that size. This accounts for $F_{t_i+1}
  + F_{t_i} -1 = F_{t_i+2}-1$ leaves. The next $F_{t_i+1}$ leaves are
  owned by $i$ as well, so at that point, at the latest, there must be
  a subtree of size $F_{t_i+1}$ of which $i$ owns all leaves. For $t_i
  \neq 1$, we can replace the root of this subtree with one input with
  arrival time $t_i$. By construction, this does not increase the
  delay.

  Here we used that $F_{t_i+1} > F_{t_i}$ to give a lower bound of the
  depth of the owned subtree. The only exception from this is the case
  $t_i = 1$, which can be treated analogously: every input with
  $t_i=1$ owns two leaves, and by similar arguments as for the subtree
  sequence, one of them must be at depth 1 in the Fibonacci tree. 

  After removing all leaves that have not been replaced by any
  original input, we obtain a prefix tree computing $z_n \circ \dots
  \circ z_1$ with delay $k-1$. All of these arguments used the
  assumption of skewed arrival times also for the inputs, which can be achieved in such a
  way that the actual delay of the circuit increases to at most $k$.
\end{proof}

\begin{figure}[!t]
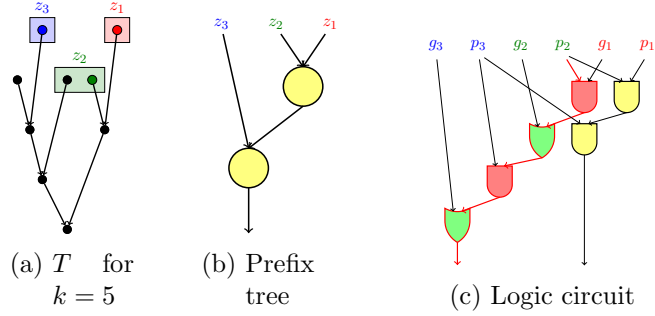

  \centering
  \begin{subfigure}[b]{.1\linewidth}
    
  \centering{
    \resizebox{1\linewidth}{!}{
      \begin{tikzpicture}
        \input{fib-t}
      \end{tikzpicture}
    }
    \caption{$T$ for $k = 5$}
    \label{fig:fib-t}%
  }

  \end{subfigure}
  \qquad
  \begin{subfigure}[b]{.12\linewidth}
    
  \centering{
    \resizebox{1\linewidth}{!}{
      \begin{tikzpicture}
        \input{fib-pfx}
      \end{tikzpicture}
    }
    \caption{Prefix tree}
    \label{fig:fib-pfx}%
  }

  \end{subfigure}
  \qquad
  \begin{subfigure}[b]{.2\linewidth}
    
  \centering{
    \resizebox{1\linewidth}{!}{
      \begin{tikzpicture}
        \input{fib-logic}
      \end{tikzpicture}
    }
    \caption{Logic circuit}
    \label{fig:fib-logic}%
  }

  \end{subfigure}
  \caption{A tight example}
  \label{fig:fib-tight} 
\end{figure}
The upper bound of $k$ is tight for the final logic circuit as
evident from the example $0, 1, 0$, where $\sum_{i=1}^3 (F_{t_i+3} -
1) = 1 + 2 + 1 = 4$, so $k = F_k = 5$ (see
Figure~\ref{fig:fib-tight}).

For this arrival time profile, the algorithm will (implicitly)
construct the Fibonacci tree $T$ and assign leaves to the inputs as in
Figure~\ref{fig:fib-t}, where the colored vertices represent the
positions at which the inputs will actually be inserted into the
tree. These do not have to be leaves in general. After deleting
redundant inputs, we obtain a prefix tree (Figure~\ref{fig:fib-pfx})
and a corresponding logic circuit (Figure~\ref{fig:fib-logic}). Note
that $p_2$ has arrival time $1$ and the red path contains four gates,
hence the logic circuit has delay $5$.

From the proof of Lemma~\ref{lem:fib-pfx}, it is easy to see how to
avoid the enumeration of all potential splitting positions
$l=1,\dots,n-1$ in (\ref{eqn:bl-dp}). Since there are $F_{k-1}$ leaves
in the left subtree of $T$ and $F_{k-2}$ in the right subtree, let $$j
= \min \left\{1 \leq j \leq n : \sum_{i = 1}^{j} (F_{t_i + 3} - 1)
  \geq F_{k-2}\right\}$$ and $f = F_{k-2} - \sum_{i = 1}^{j-1} (F_{t_i
  + 3} - 1)$, then $f$ counts how many leaves belonging to input $j$
are part of the right subtree, and $j$ is the only input that might
have leaves in both subtrees. Since in our decomposition the leftmost
$F_{t_j+1}$ leaves of the right subtree belong to a Fibonacci tree of
size $F_{t_j+1}$, $j$ should be on the right side of the decomposition
if and only if $f \geq F_{t_j+1}$. Otherwise, there are at least
$F_{t_j+2}$ leaves on the left side, hence in our sequence of subtrees
$j$ might own all leaves of a subtree of size $F_{t_j}$, but the
remaining leaves must belong to and cover a subtree of size
$F_{t_j+1}$, hence $j$ should be on the left side. Note that it is
never optimal to assign all leaves to the same side, thus this
partition can always be assumed as proper without increasing the
delay. After updating the number of leaves belonging to $j$ on the
side it is assigned to, this yields a recursive procedure that
terminates when there is only one index left for a subtree.

\begin{lemma}\label{lem:fast-bl}
  Given input arrival times $t_1, \dots, t_n \in \mathbb{N}_0$, let
  $F_{k}$ be the first Fibonacci number that is at least as large as
  $\sum_{i=1}^{n} (F_{t_i + 3} - 1)$. A prefix tree for these input
  arrival times with delay at most $k$ can be found with running time
  $\mathcal{O}\left(n \log n + k + \max t_i\right)$ under the
  assumption that we can perform additions and multiplications by a
  constant on numbers of arbitrary size in constant time (an
  assumption we will show how to avoid later). If $t_i \in
  \mathcal{O}(n)$ for all $i$, then the running time is
  $\mathcal{O}(n\log n)$.
\end{lemma}

\begin{proof} %
  We have already argued that the algorithm achieves the stated delay
  bound. We show that this partitioning strategy will ensure that
  every input $i$ is substituted for a subtree of size at least
  $F_{t_i+1}$.

  If there is only one index $i$ remaining, it was either the
  rightmost (lowest) or leftmost (highest) index in the previous
  step. If it was the rightmost index, then the subtree previously
  contained $F_{t_i+2}$ of its leaves as well as at least one more
  leaf, hence $k \geq t_i + 3$ and the right subtree has size at least
  $F_{t_i+1}$, so replacing this subtree by input $i$ leads to the
  claimed delay by the argument used in Lemma~\ref{lem:fib-pfx}. If it
  was the leftmost index, a similar argument applies.

  For the running time estimate, we compute the indices assigned to
  every leaf and the delay bound $k$ in time $\mathcal{O}\left(
    n + k + \max t_i \right)$. There are $n-1$ recursive
  partitioning steps, during each of which we find the input $j$ as
  the input index to own leaves in the left subtree. This can be done
  in logarithmic time using binary search in the sorted array of the
  indices of the first leaf belonging to every input.
\end{proof}
\begin{figure}[htb]
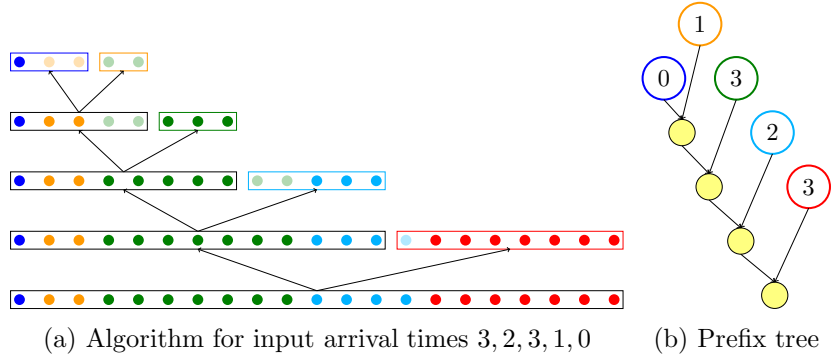

  \centering
  \begin{subfigure}[b]{0.5\linewidth}
    
  \centering{
    \resizebox{1\linewidth}{!}{
      \begin{tikzpicture}
        \input{core-ex}
      \end{tikzpicture}
    }
    \caption{Algorithm for input arrival times $3, 2, 3, 1, 0$}
    \label{fig:core-ex}%
  }

  \end{subfigure}
  \begin{subfigure}[b]{0.16\linewidth}
    
  \centering{
    \resizebox{1\linewidth}{!}{
      \begin{tikzpicture}
        \input{core-tree}
      \end{tikzpicture}
    }
    \caption{Prefix tree}
    \label{fig:core-tree}%
  }

  \end{subfigure}
  \caption{Example of the algorithm}
  \label{fig:alg-ex}
\end{figure}

Figure~\ref{fig:core-ex} shows how the algorithm works for the
sequence of input arrival times $3, 2, 3, 1, 0$. The number of leaves
we need is $7 + 4 + 7 + 2 + 1 = 21$, therefore $k=8$ suffices. We
number the leaves from right to left as $1, \dots, 21$. After the
first split, $3$ light blue leaves are in the left subtree, hence the
corresponding input is assigned to the left subtree.  Note that for
the orange leaves, we end up assigning them to a subtree that does not
contain any orange leaves in the beginning in order to ensure a proper
partition. We obtain the result shown in Figure~\ref{fig:core-tree}.

\begin{lemma} \label{lem:reduction}
  We can construct a prefix tree with the delay bound of
  Lemma~\ref{lem:fib-pfx} for any instance $(t_1, \dots, t_n)$ by
  constructing a prefix tree for an instance $(t'_1, \dots, t'_n)$
  with $\max t'_i \leq 2n - 1 $.
\end{lemma}
This follows from the fact that the longest path from any input to the
output contains at most $n-1$ prefix gates. The maximum delay
difference can be assumed as $2n-2$, since any input with earlier
arrival time will never be critical.

\begin{theorem} \label{thm:my-pfx} For $n$ inputs with arrival times
    $t_1, \dots, t_n \in \mathbb{N}_0$, the algorithm finds a prefix
    carry bit circuit for $c_{n+1}$ with $$\delay(c_{n+1}) \leq k \leq
    \floor{\loq \left(\sum_{i=1}^{n} \varphi^{t_i}\right)} + 4.$$ 
    The constructed logic circuit has size at most $3n - 3$ and
    maximum fan-out two at all logic gates and inputs. Furthermore,
    the delay is at most
    $$k \leq \loq \left(\sum_{i=1}^n 2^{t_i}\right) + 2.673 \leq 1.441 \ld
    \left(\sum_{i=1}^n 2^{t_i}\right) + 2.673.$$
  \end{theorem}
  \begin{proof}
  The size and fan-out bounds follow from Lemma~\ref{lem:small-bl}.
  The delay of the
  constructed circuit is $k$. By choice of $k$, we know that
  $\sum_{i=1}^n (F_{t_i+3} -1) \geq F_{k-1} + 1$. With $\varphi = \frac{1+
    \sqrt{5}}{2}$, $\psi = \frac{1-\sqrt{5}}{2}$ and the exact formula
  $F_n = \frac{1}{\sqrt{5}} \cdot (\varphi^n - \psi^n)$, it follows that
  $|\sqrt{5}F_n - \varphi^n| \leq 1$ and for $n\geq 1$,
  $|\sqrt{5}F_n - \varphi^n| \leq |\psi|$.

  Now $k -1 = 0$ can only be true if there is only one input. In this
  case, the stated delay bound is trivially true. Otherwise, we obtain
  the estimate:
  \begin{align*}
    &k-1 = \loq \left(\varphi^{k-1}\right) \leq \loq \left(\sqrt{5}
      \left(F_{k-1} + 1\right)\right) \leq \loq \left(\sqrt{5}
      \left(\sum_{i=1}^{n} \left(F_{t_i + 3}  - 1\right)\right)\right) \\
    &\leq \loq \left(\sqrt{5} \left(\sum_{i=1}^{n}
        \left(\frac{1}{\sqrt{5}} \left(\varphi^{t_i + 3} + 1\right) - 1\right)
      \right)\right) \leq \loq \left(\sum_{i=1}^{n} \varphi^{t_i + 3}
    \right) = \loq \left(\sum_{i=1}^{n} \varphi^{t_i} \right) + 3,
  \end{align*}
  which proves the first claim. 

  For a single input, the second delay bound is trivially
  true.  Furthermore, for $t_i \geq 0$, $F_{t_i + 3} - 1\leq
  2^{t_i}$. We obtain the estimate:
  \begin{align*}
    &k-1 = \loq \left(\varphi^{k-1}\right) \leq \loq \left(\sqrt{5}
      \left(F_{k-1} + 1\right)\right) \leq \loq \left(\sqrt{5}
      \left(\sum_{i=1}^{n} \left(F_{t_i + 3}  - 1\right)\right)\right) \\
    &  \leq \loq \left( \sqrt{5} \left (\sum_{i=1}^n 2^{t_i}\right)\right)
    = \loq \left(\sum_{i=1}^n 2^{t_i}\right) + \loq \sqrt{5} 
    \leq 1.441\ld \left(\sum_{i=1}^n 2^{t_i}\right) + 1.673. 
  \end{align*}
\end{proof}

Our proof allows an improvement over the delay bound of \cite{bonn1}
due to a refined analysis.  A running time of $\mathcal{O}(n\log n)$
follows from Lemma~\ref{lem:fast-bl} assuming that we can add numbers
of linear size and multiply them by a constant in constant time.
Under the more practical assumption that these operations take linear
time with respect to the number of digits, the algorithm has
super-quadratic running time, which can be avoided as follows:

\begin{theorem} \label{thm:fast-pfx} For any fixed $\gamma > 1$, a
  prefix carry bit circuit as in the setting of
  Theorem~\ref{thm:my-pfx} with $$\delay(c_{n+1}) \leq \loq
  \left(\sum_{i=1}^{n} \varphi^{t_i}\right) + 4 + 2.1 \cdot
  n^{1-\gamma }$$ can be found in $\mathcal{O}(n\gamma\log^2 n)$ time
  assuming linear-time addition and multiplication with constants.  It
  satisfies
  $$\delay(c_{n+1}) \leq 1.441 \ld
  \left(\sum_{i=1}^n 2^{t_i}\right) + 2.673 + 2.1 \cdot
  n^{1-1.4\gamma}.$$
\end{theorem}
\begin{proof}
  By Lemma~\ref{lem:reduction} and Theorem~\ref{thm:my-pfx}, we can
  solve instances with $\max t_i - \min t_i \leq \gamma \ceil{\loq n}$
  in $\mathcal{O}(n \gamma\log^2 n)$ time with linear-time addition.

  Given an instance $t_1, \dots, t_n$, we set $t_i' = \max\sset{t_i,
    \max_{j \in \sset{1, \dots, n}}
    t_j - \gamma\ceil{\loq n}}$ and compute a circuit for the modified
  instance in $\mathcal{O}(n \gamma\log^2 n)$. When reverting to the original
  arrival times, the delay of this solution does not increase, because
  none of the arrival times do. Therefore, 
  \begin{align*}
    \delay(c_{n+1}) - 4&\leq \loq \left(\sum_{i=1}^{n} \phi^{t_i'}\right) 
\\    &\leq \loq \left(n \cdot \phi^{\max t_i - \gamma\ceil{\loq n}}
      + \sum_{i=1}^{n} \phi^{t_i}\right) \\
    &\leq \loq \left(\phi^{\max t_i + (1-\gamma) \loq n}
      + \sum_{i=1}^{n} \phi^{t_i}\right) \\
    &\leq \loq \left(\sum_{i=1}^{n} \phi^{t_i}\right) + \loq \left(1 +
    \phi^{(1-\gamma) \loq n}\right) \\
    &\leq \loq \left(\sum_{i=1}^{n} \phi^{t_i}\right) + 
    \loq e \cdot n^{1-\gamma }, 
  \end{align*}
  and $\loq e < 2.1$. For the dual logarithm-based delay bound, we
  have
  \begin{align*}
    \delay(c_{n+1}) - 2.673 &\leq \loq 2 \cdot\ld \left(\sum_{i=1}^{n} 2^{t_i'}\right)
    \\ &\leq \loq 2 \cdot\ld \left(n \cdot 2^{\max t_i - \gamma\ceil{\loq
          n}}
      + \sum_{i=1}^{n} 2^{t_i}\right) \\
    &\leq\loq 2 \cdot \ld \left(2^{\max t_i - (1-\gamma\loq 2) \ld n}
      + \sum_{i=1}^{n} 2^{t_i}\right) \\
    &\leq \loq 2  \left(\ld \left(\sum_{i=1}^{n} 2^{t_i}\right) + \ld
    \left(1 +
      2^{(1-\gamma\loq 2) \ld n}\right)\right) \\
    &\leq \loq 2 \cdot \ld \left(\sum_{i=1}^{n} 2^{t_i}\right) + (\loq e)
    n^{1-\gamma\loq 2},
  \end{align*}
  and $1.4 < \loq 2 < 1.441$. 
\end{proof}

For $\gamma > 1$, the additional error decreases with growing $n$. Since the algorithm
is only useful if $n \geq 2$, choosing a sufficiently large constant
$\gamma$ yields the delay bound $1.441 \ld \left(\sum_{i=1}^n
  2^{t_i}\right) + 2.674$ with running time $\mathcal{O}(n \log^2
n)$. 

\section{Algorithm for Prefix Adder Circuits}\label{sec:bl-adder}
The na\"{\i}ve parallel prefix graph construction, in which all carry
bits are computed separately by a carry bit circuit, might contain a quadratic
number of gates. Therefore, Rautenbach et al.\  also developed a parallel prefix graph construction
computing all carry bits \cite{bonn2}.

\begin{theorem}[{\cite{bonn2}}] \label{thm:bl-ppfx} Given $n \in
  \mathbb{N}$ and arrival times $t_1, \dots, t_n \in \mathbb{N}_0$, there is a
  parallel prefix graph for $n$ inputs of size $\mathcal{O}(n \log
  \log n)$ with logic delay $2\log_2\left(\sum_{i=1}^n 2^{t_i}\right) + 6 \ld \ld n + \mathcal{O}(1).$
\end{theorem}

The primary objective in \cite{bonn2} is to minimize the delay of the
prefix graph instead of the underlying logic circuit.  We will improve
the performance guarantee for a similar construction as in
\cite{bonn2} by using a carry bit circuit as in
Section~\ref{sec:bl-core} as a subroutinte.

Given inputs $t_1, \dots, t_n$, we partition the set $\{1, \dots, n\}$
into $l = \ceil{\sqrt{n}}$ subsets $V_1, \dots, V_l$, each containing
$l$ or $l-1$ consecutive indices. Let $Z_i = \circ_{j \in V_i} z_j$,
where $Z_i$ is computed by a circuit constructed by the carry bit
algorithm. This is shown in green, labeled ``Best'', in
Figure~\ref{fig:bl-pfx}. The parallel prefix graph construction is
applied recursively to compute prefixes for all groups without their
highest index as well as for the $l-1$ inputs $Z_1, \dots, Z_{l-1}$
(which corresponds to the red boxes labeled ``Recursion'' in
Figure~\ref{fig:bl-pfx}), i.\ e.\ we build $l+1$ parallel prefix
graphs, each with at most $l-1$ inputs. As a final step, we combine
all prefixes from group $i$ with the $(i-1)$-th prefix of the $Z_i$
and add one more prefix gate combining $Z_l$ with the $(l-1)$-th
prefix of the $Z_i$. This yields a parallel prefix graph.

  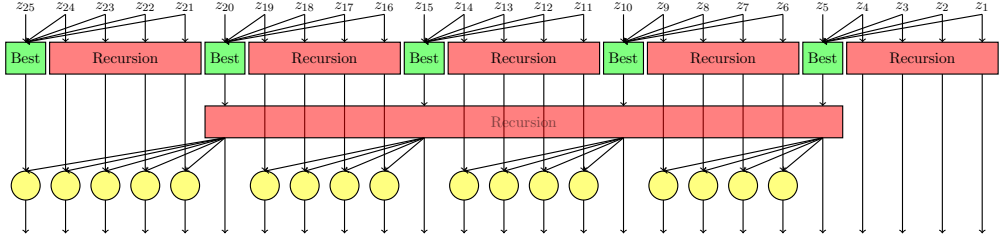
\begin{figure}[!tb]%
    \centering{%
      \resizebox{0.8\linewidth}{!}{%
        \begin{tikzpicture}
          \node at (1, 6.1) {$z_{25}$};
\node at (2, 6.1) {$z_{24}$};
\node at (3, 6.1) {$z_{23}$};
\node at (4, 6.1) {$z_{22}$};
\node at (5, 6.1) {$z_{21}$};
\node at (6, 6.1) {$z_{20}$};
\node at (7, 6.1) {$z_{19}$};
\node at (8, 6.1) {$z_{18}$};
\node at (9, 6.1) {$z_{17}$};
\node at (10, 6.1) {$z_{16}$};
\node at (11, 6.1) {$z_{15}$};
\node at (12, 6.1) {$z_{14}$};
\node at (13, 6.1) {$z_{13}$};
\node at (14, 6.1) {$z_{12}$};
\node at (15, 6.1) {$z_{11}$};
\node at (16, 6.1) {$z_{10}$};
\node at (17, 6.1) {$z_{9}$};
\node at (18, 6.1) {$z_{8}$};
\node at (19, 6.1) {$z_{7}$};
\node at (20, 6.1) {$z_{6}$};
\node at (21, 6.1) {$z_{5}$};
\node at (22, 6.1) {$z_{4}$};
\node at (23, 6.1) {$z_{3}$};
\node at (24, 6.1) {$z_{2}$};
\node at (25, 6.1) {$z_{1}$};

\node[outer sep=0pt, draw, fill=red!50, thick, rectangle, scale=1, minimum
width=3.8cm, minimum height=0.8cm] at (3.5,4.8)(g2){Recursion};
\node[outer sep=0pt, draw, fill=green!50, thick, rectangle, scale=1, minimum
height=0.8cm, minimum width=0.8cm] at (1,4.8)(g1){Best};
\node[outer sep=0pt, draw, fill=red!50, thick, rectangle, scale=1, minimum
width=3.8cm, minimum height=0.8cm] at (8.5,4.8)(g4){Recursion};
\node[outer sep=0pt, draw, fill=green!50, thick, rectangle, scale=1, minimum
height=0.8cm, minimum width=0.8cm] at (6,4.8)(g3){Best};
\node[outer sep=0pt, draw, fill=red!50, thick, rectangle, scale=1, minimum
width=3.8cm, minimum height=0.8cm] at (13.5,4.8)(g6){Recursion};
\node[outer sep=0pt, draw, fill=green!50, thick, rectangle, scale=1, minimum
height=0.8cm, minimum width=0.8cm] at (11,4.8)(g5){Best};
\node[outer sep=0pt, draw, fill=red!50, thick, rectangle, scale=1, minimum
width=3.8cm, minimum height=0.8cm] at (18.5,4.8)(g8){Recursion};
\node[outer sep=0pt, draw, fill=green!50, thick, rectangle, scale=1, minimum
height=0.8cm, minimum width=0.8cm] at (16,4.8)(g7){Best};
\node[outer sep=0pt, draw, fill=red!50, thick, rectangle, scale=1, minimum
width=3.8cm, minimum height=0.8cm] at (23.5,4.8)(g10){Recursion};
\node[outer sep=0pt, draw, fill=green!50, thick, rectangle, scale=1, minimum
height=0.8cm, minimum width=0.8cm] at (21,4.8)(g9){Best};

\path[->,thick] (1,5.9) edge (1,5.2);
\path[->,thick] (2,5.9) edge (2,5.2);
\path[->,thick] (3,5.9) edge (3,5.2);
\path[->,thick] (4,5.9) edge (4,5.2);
\path[->,thick] (5,5.9) edge (5,5.2);
\path[->,thick] (6,5.9) edge (6,5.2);
\path[->,thick] (7,5.9) edge (7,5.2);
\path[->,thick] (8,5.9) edge (8,5.2);
\path[->,thick] (9,5.9) edge (9,5.2);
\path[->,thick] (10,5.9) edge (10,5.2);
\path[->,thick] (11,5.9) edge (11,5.2);
\path[->,thick] (12,5.9) edge (12,5.2);
\path[->,thick] (13,5.9) edge (13,5.2);
\path[->,thick] (14,5.9) edge (14,5.2);
\path[->,thick] (15,5.9) edge (15,5.2);
\path[->,thick] (16,5.9) edge (16,5.2);
\path[->,thick] (17,5.9) edge (17,5.2);
\path[->,thick] (18,5.9) edge (18,5.2);
\path[->,thick] (19,5.9) edge (19,5.2);
\path[->,thick] (20,5.9) edge (20,5.2);
\path[->,thick] (21,5.9) edge (21,5.2);
\path[->,thick] (22,5.9) edge (22,5.2);
\path[->,thick] (23,5.9) edge (23,5.2);
\path[->,thick] (24,5.9) edge (24,5.2);
\path[->,thick] (25,5.9) edge (25,5.2);
\path[->,thick] (2,5.9) edge (1,5.2);
\path[->,thick] (3,5.9) edge (1,5.2);
\path[->,thick] (4,5.9) edge (1,5.2);
\path[->,thick] (5,5.9) edge (1,5.2);
\path[->,thick] (7,5.9) edge (6,5.2);
\path[->,thick] (8,5.9) edge (6,5.2);
\path[->,thick] (9,5.9) edge (6,5.2);
\path[->,thick] (10,5.9) edge (6,5.2);
\path[->,thick] (12,5.9) edge (11,5.2);
\path[->,thick] (13,5.9) edge (11,5.2);
\path[->,thick] (14,5.9) edge (11,5.2);
\path[->,thick] (15,5.9) edge (11,5.2);
\path[->,thick] (17,5.9) edge (16,5.2);
\path[->,thick] (18,5.9) edge (16,5.2);
\path[->,thick] (19,5.9) edge (16,5.2);
\path[->,thick] (20,5.9) edge (16,5.2);
\path[->,thick] (22,5.9) edge (21,5.2);
\path[->,thick] (23,5.9) edge (21,5.2);
\path[->,thick] (24,5.9) edge (21,5.2);
\path[->,thick] (25,5.9) edge (21,5.2);

\node[outer sep=0pt, draw, fill=yellow!50, thick, circle, scale=2] at (20,1.6)(n20){};
\node[outer sep=0pt, draw, fill=yellow!50, thick, circle, scale=2] at (19,1.6)(n19){};
\node[outer sep=0pt, draw, fill=yellow!50, thick, circle, scale=2] at (18,1.6)(n18){};
\node[outer sep=0pt, draw, fill=yellow!50, thick, circle, scale=2] at (17,1.6)(n17){};
\node[outer sep=0pt, draw, fill=yellow!50, thick, circle, scale=2] at (15,1.6)(n15){};
\node[outer sep=0pt, draw, fill=yellow!50, thick, circle, scale=2] at (14,1.6)(n14){};
\node[outer sep=0pt, draw, fill=yellow!50, thick, circle, scale=2] at (13,1.6)(n13){};
\node[outer sep=0pt, draw, fill=yellow!50, thick, circle, scale=2] at (12,1.6)(n12){};
\node[outer sep=0pt, draw, fill=yellow!50, thick, circle, scale=2] at (10,1.6)(n10){};
\node[outer sep=0pt, draw, fill=yellow!50, thick, circle, scale=2] at (9,1.6)(n9){};
\node[outer sep=0pt, draw, fill=yellow!50, thick, circle, scale=2] at (8,1.6)(n8){};
\node[outer sep=0pt, draw, fill=yellow!50, thick, circle, scale=2] at (7,1.6)(n7){};
\node[outer sep=0pt, draw, fill=yellow!50, thick, circle, scale=2] at (5,1.6)(n5){};
\node[outer sep=0pt, draw, fill=yellow!50, thick, circle, scale=2] at (4,1.6)(n4){};
\node[outer sep=0pt, draw, fill=yellow!50, thick, circle, scale=2] at (3,1.6)(n3){};
\node[outer sep=0pt, draw, fill=yellow!50, thick, circle, scale=2] at (2,1.6)(n2){};
\node[outer sep=0pt, draw, fill=yellow!50, thick, circle, scale=2] at (1,1.6)(n1){};

\path[->,thick] (6,4.4) edge (6,3.6);
\path[->,thick] (6,2.8) edge (6,0.4);
\path[->,thick] (6,2.8) edge (n5.north);
\path[->,thick] (6,2.8) edge (n4.north);
\path[->,thick] (6,2.8) edge (n3.north);
\path[->,thick] (6,2.8) edge (n2.north);
\path[->,thick] (6,2.8) edge (n1.north);
\path[->,thick] (11,4.4) edge (11,3.6);
\path[->,thick] (11,2.8) edge (11,0.4);
\path[->,thick] (11,2.8) edge (n10.north);
\path[->,thick] (11,2.8) edge (n9.north);
\path[->,thick] (11,2.8) edge (n8.north);
\path[->,thick] (11,2.8) edge (n7.north);
\path[->,thick] (16,4.4) edge (16,3.6);
\path[->,thick] (16,2.8) edge (16,0.4);
\path[->,thick] (16,2.8) edge (n15.north);
\path[->,thick] (16,2.8) edge (n14.north);
\path[->,thick] (16,2.8) edge (n13.north);
\path[->,thick] (16,2.8) edge (n12.north);
\path[->,thick] (21,4.4) edge (21,3.6);
\path[->,thick] (21,2.8) edge (21,0.4);
\path[->,thick] (21,2.8) edge (n20.north);
\path[->,thick] (21,2.8) edge (n19.north);
\path[->,thick] (21,2.8) edge (n18.north);
\path[->,thick] (21,2.8) edge (n17.north);

\path[->,thick] (1,4.4) edge (n1.north);
\path[->,thick] (n1.south) edge (1,0.4);
\path[->,thick] (2,4.4) edge (n2.north);
\path[->,thick] (n2.south) edge (2,0.4);
\path[->,thick] (3,4.4) edge (n3.north);
\path[->,thick] (n3.south) edge (3,0.4);
\path[->,thick] (4,4.4) edge (n4.north);
\path[->,thick] (n4.south) edge (4,0.4);
\path[->,thick] (5,4.4) edge (n5.north);
\path[->,thick] (n5.south) edge (5,0.4);
\path[->,thick] (7,4.4) edge (n7.north);
\path[->,thick] (n7.south) edge (7,0.4);
\path[->,thick] (8,4.4) edge (n8.north);
\path[->,thick] (n8.south) edge (8,0.4);
\path[->,thick] (9,4.4) edge (n9.north);
\path[->,thick] (n9.south) edge (9,0.4);
\path[->,thick] (10,4.4) edge (n10.north);
\path[->,thick] (n10.south) edge (10,0.4);
\path[->,thick] (12,4.4) edge (n12.north);
\path[->,thick] (n12.south) edge (12,0.4);
\path[->,thick] (13,4.4) edge (n13.north);
\path[->,thick] (n13.south) edge (13,0.4);
\path[->,thick] (14,4.4) edge (n14.north);
\path[->,thick] (n14.south) edge (14,0.4);
\path[->,thick] (15,4.4) edge (n15.north);
\path[->,thick] (n15.south) edge (15,0.4);
\path[->,thick] (17,4.4) edge (n17.north);
\path[->,thick] (n17.south) edge (17,0.4);
\path[->,thick] (18,4.4) edge (n18.north);
\path[->,thick] (n18.south) edge (18,0.4);
\path[->,thick] (19,4.4) edge (n19.north);
\path[->,thick] (n19.south) edge (19,0.4);
\path[->,thick] (20,4.4) edge (n20.north);
\path[->,thick] (n20.south) edge (20,0.4);
\path[->,thick] (22,4.4) edge (22,0.4);
\path[->,thick] (23,4.4) edge (23,0.4);
\path[->,thick] (24,4.4) edge (24,0.4);
\path[->,thick] (25,4.4) edge (25,0.4);

\node[outer sep=0pt, draw, fill=red, fill opacity=.5, text opacity = 1, thick, rectangle, scale=1, minimum
width=16cm, minimum height=0.8cm] at (13.5,3.2)(g2){Recursion};
        \end{tikzpicture}%
      }%
      \caption{Prefix graph construction}%
      \label{fig:bl-pfx}%
    }%
  \end{figure}%

The following two lemmas analyze the size of the resulting parallel prefix graph
and the running time of its construction.

\begin{lemma} \label{lem:bl-ppfx-size-app}
  The parallel prefix graph in \cite{bonn2} and the modified
  construction above  have the same size; for
  $n\geq 3$, it is bounded by $2n\ld \ld n$ in terms of prefix gates
  and $6n \ld \ld n$ in terms of logic gates.
\end{lemma}
\begin{proof}
  Consider Figure~\ref{fig:sqrt} and proceed by induction on the
  number of inputs. On a level with $n$ inputs, the number of green
  gates and the number of yellow gates are both at most $n$. The total
  number of inputs of recursion blocks can be bounded by $n$ as well:
  if there are $l$ groups, then $n-l$ original inputs are inputs of
  recursion blocks; one further recursion block has $l-1$ inputs. For
  small $n$, the correctness follows from Figure~\ref{fig:sqrt}, e.\
  g.\ for $n = 3$, the size bound is $5$ and $3$ gates are required.

  Let $V_1, \dots, V_l$ be the groups, $l\geq 3$, then by induction
  hypothesis, the prefix gate size is bounded by
  \begin{align*}
    &\phantom{= } 2n + 2(l-1) \ld \ld (l-1) + \sum_{i=1}^l 2(|V_i|-1)
    \ld \ld (|V_i|-1) \\
    &\leq 2n + \sum_{i=1}^l 2|V_i| \ld \ld (l-1) \\
    &\leq 2n + 2n \ld \ld \sqrt{n} = 2 n \ld \ld n.
  \end{align*}
  For logic gates, the size increases by a factor of three. 
\end{proof}

\begin{lemma} \label{lem:rt-app}
  The parallel prefix graph above can be
  computed in $\mathcal{O}(n \log^2 n)$ time. 
\end{lemma}
\begin{proof}
  As in Theorem~\ref{thm:fast-pfx}, we round all running times up to
  at least $\max t_i - \gamma \ceil{\ld n}$ for a fixed $\gamma >
  1$. For this arrival time profile, we have already shown that
  $$1.441\left(\sum_{i=1}^n 2^{t_i'}\right)  \leq  1.441 \ld
  \left(\sum_{i=1}^n 2^{t_i}\right) + 2.1 \cdot n^{1-1.4\gamma}.$$
  Therefore, we can use the rounded arrival time profile to achieve a
  delay guarantee of 
  $$1.441 \ld
  \left(\sum_{i=1}^n 2^{t_i}\right) + 5 \ld \ld n + 4.5$$ for $\gamma
  = 3$. This means that all numbers in the computations have size at
  most $\mathcal{O}(\log n)$, thus it remains to bound the number of
  operations by $\mathcal{O}(n \log n)$. 

  For each level $l = 1, \dots, \ld \ld n$ of the recursion, we have a
  partition of the $n$ inputs into groups, where the maximum group
  size is bounded by $n^{1/2^l}$. Therefore, the prefix trees for the
  $Z_i$ can be computed in $\mathcal{O}\left(n \log
    \left(n^{1/2^l}\right)\right) =
  \mathcal{O}\left(\left(\frac{1}{2}\right)^l n\log n\right)$
  time. All $Z_i$ require time $\mathcal{O}(n \log n)$, because this
  is a geometric series. All remaining gates are prefix gates; they
  have fixed positions, thus each of them requires only constant time
  to compute, and there are $\mathcal{O}(n \log \log n)$ such gates in
  total.
\end{proof}

The new parallel prefix graph construction is summarized in the
following theorem.
\begin{theorem} \label{thm:my-ppfx} Given $n \in \mathbb{N}$ and
  arrival times $t_1, \dots, t_n \in \mathbb{N}_0$, our algorithm
  finds a parallel prefix graph with logic gate delay at most
$$\log_{\varphi} \left(\sum_{i=1}^n \varphi^{t_i}\right) + 5 \log_2 \log_2 n +
4.5 \le 1.441 \ld  \left(\sum_{i=1}^n 2^{t_i}\right) + 5 \log_2 \log_2 n +
4.5.$$
It can be implemented with running time $\mathcal{O}(n \log^2 n)$ and
the computed circuit has size at most $6 n \ld \ld n$ in terms of logic gates. 
\end{theorem}

This theorem implies that for $n$ sufficiently large, we have a
$1.441$-approximation algorithm in terms of the delay for a prefix
adder. The algorithm of \cite{bonn2} has a running time of
$\Omega(n^2)$, which the use of our carry bit algorithm improves to a
near-linear running time, even with linear-time addition.

To prove the delay bound, we assume that all arrival
times are skewed by one time unit. Under this assumption, let $w =
\sum_{i=1}^n \varphi^{t_i}$, and let $\delay(w, n)$ denote the maximum
delay for a circuit constructed as above with $n\geq 3$ inputs and an
arrival time profile leading to the same $w$. Then $\delay(w, n)+1$ is
an upper bound on the delay of the constructed circuit, and we have:

\begin{lemma} \label{lem:delay-bl-pfx} For $n$ input pairs with skewed arrival
  times $t_1, \dots, t_n$, let $w = \sum_{i=1}^n
  \varphi^{t_i}$. Then we have
  $$\delay(w, n) \leq \loq w + 5\ld\ld n + 3.$$
\end{lemma}
\begin{proof}
  We may assume that the given arrival time profile achieves the
  maximum delay, i.\ e.\ for $t_1, \dots, t_n$, the construction
  actually has a delay of $\delay(w, n)$.
 
  By Theorem~\ref{thm:my-pfx} and using the assumption that the
  propagate signals arrive earlier than the generate signals, we can
  compute $Z_i$ with delay $\loq \left(\sum_{j\in V_i}
    \varphi^{t_j}\right) + 3$. Therefore, their prefix graph has delay
  at most $$\delay\left(\sum_{i=1}^{l} \varphi^{\loq \left(\sum_{j\in
          V_i} \varphi^{t_j}\right) + 3}, \ceil{\sqrt{n}}-1\right)
  =\delay\left(\varphi^3 \cdot \sum_{j=1}^{n} \varphi^{t_j},
    \ceil{\sqrt{n}}-1\right).$$ For each of the groups $V_i$
  containing $\ceil{\sqrt{n}}$ or $\ceil{\sqrt{n}}-1$ inputs, the
  prefix graph of all but its last input (highest index) has delay at
  most $$\delay(w, \ceil{\sqrt{n}}-1) \leq \delay(\varphi^3 w,
  \ceil{\sqrt{n}} -1 ),$$ as $\delay(w,n)$ is monotonically increasing
  in $w$ and $n$. Therefore, the combination of a prefix of one of the
  $V_i$ and the corresponding prefix of the $Z_i$ has logic gate delay
  at most $\delay(w, n) \leq \delay\left(\varphi^3 w,
    \ceil{\sqrt{n}}-1\right)+2$. 

  We prove the absolute delay estimate by induction on $n$. For $n
  \leq 3$, $\delay(w, n) - \log_{\varphi} w \leq 3$ for all input
  sequences with this parameter $w$ as $\loq w\geq \max_i
  t_i$. Therefore, for $n \geq 4$, $\delay(w,n)$ is bounded by
  \begin{align*}
    \delay(\varphi^3w, \ceil{\sqrt{n}}-1) + 2
    &\leq \loq (\varphi^3 w) + 5\ld \ld(\sqrt{n}) + 5
    \leq \loq w + 5 \ld(0.5 \ld n) + 8 \\
    &= \loq w + 5\ld\ld n - 5 + 8 
    = \loq w + 5\ld\ld n + 3.
\end{align*}
\end{proof}

Without assuming skewed arrival times, we achieve a delay bound
of $\loq w + 5 \ld \ld n + 4$. For $n=25$, an example is shown in
Figure~\ref{fig:sqrt}. Gates are colored by the part of the recursion
they represent; in this special case, some gates can be used to compute
the $Z_i$ as well as the group prefixes, hence they are both red and
green.

  \begin{figure}[!tb]%
    \centering{%
      \resizebox{0.8\linewidth}{!}{%
        \begin{tikzpicture}
          \input{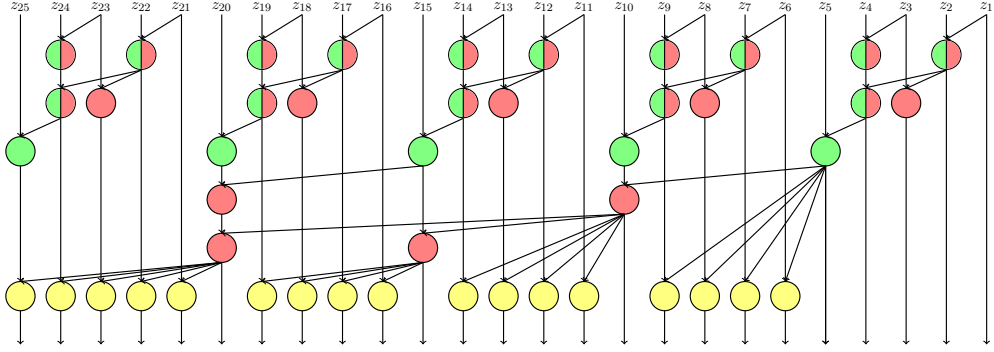}
        \end{tikzpicture}%
      }%
      \caption{Parallel prefix graph for uniform arrival times}%
      \label{fig:sqrt}%
    }%
  \end{figure}%

The construction in \cite{bonn2} and our variant of it both have a
very high fan-out; for $n$ inputs, the fan-out is at least
$\ceil{\sqrt{n}}$. In a physical implementation such a high fanout
induces a significant delay and requires the insertion of
duplicate gates into the interconnect to repeat the signals. The high
fan-outs occur precisely at the $Z_i$-prefixes, therefore they
accumulate on a critical path. For $n$ inputs, the fan-out can be
redistributed to duplicate gates with fan-out 2 using depth
$\frac{1}{2} \ceil{\ld n} + 1$; this will lead to an overall increase
in delay of $\ceil{\ld n} + \mathcal{O}(\ld \ld n)$ for a given path
\cite{bonn2}. Therefore, we obtain a $2.441$-approximation algorithm if
the fan-out is bounded by $2$, improving the $3$-approximation
achieved by \cite{bonn2} in this scenario.

\section{A Lower Bound for Prefix Adders} \label{sec:lb}
Lemma~\ref{lem:pfx-to-logic} shows that a lower bound for the delay of
a prefix tree for a single carry bit is given by an optimal binary
tree with depth one for the left child and depth two for the right
child in which the leaves represent inputs and their right-to-left
order corresponds to the ordering of the inputs. For zero arrival
times, this is achieved by a Fibonacci tree. Rautenbach et al.\
\cite{code-trees} observed that this a special case of a more general
concept: alphabetic code trees with unequal letter costs. These can be
used to obtain general lower bounds, which we improve and state
explicitly by using the specific properties of our application.

\begin{lemma}\label{lem:bl-lb}
  Given $n$ inputs with integral arrival times $t_1, \dots, t_n\in
  \mathbb{N}_0$, a prefix tree computing their carry bit $c_{n+1}$ has
  logic gate delay at least
  $$\delay(c_{n+1}) \geq \loq \left(\sum_{i=1}^n \varphi^{t_i}\right) - 1.$$
\end{lemma}
\begin{proof}
  In Section~\ref{sec:bl-core} we saw that $F_{t_i+1}$ inputs with
  arrival time zero can be combined with depth $t_i$. Therefore, an
  optimal prefix tree for inputs with arrival times $t_1, \dots, t_n$
  of delay $k$ can be restructured into a prefix tree with
  $\sum_{i=1}^n F_{t_i+1}$ inputs with depth $k$ by replacing input
  $i$ by a Fibonacci tree for $t_i + 1$. If there is only one input,
  the lemma is trivially true, thus we may assume $\sum_{i=1}^n
  F_{t_i+1}\geq 2$. But a tree of depth $k$ has at most $F_{k+1}$
  leaves, hence $k \geq 2$ and
\begin{align*}
k+1 = \loq \left(\varphi^{k+1}\right) 
    &\geq \loq \left(\sqrt{5} \left(F_{k+1} +
      \frac{\psi^3}{\sqrt{5}}\right)\right)
    \geq \loq \left(\sqrt{5} \left(\sum_{i=1}^n F_{t_i+1} +
      \frac{\psi^3}{\sqrt{5}}\right)\right)\\
    &\geq \loq \left(\sum_{i=1}^n \left(\varphi^{t_i+1} - 
       \psi^2 + \psi^3 \right)\right)
    = \loq \left(\sum_{i=1}^n \left(\varphi^{t_i+1} - 
       \varphi\psi^2 \right)\right)\\
    &\geq \loq \left(\varphi(1 - \psi^2) \cdot \sum_{i=1}^n \varphi^{t_i} \right)
    = \loq \left(\sum_{i=1}^n \varphi^{t_i} \right),
\end{align*}
and $k \geq \loq \left(\sum_{i=1}^n \varphi^{t_i} \right) - 1$ as
claimed. 
\end{proof}
This lemma shows that  the  single carry bit circuits in Section~\ref{sec:bl-core} 
have  optimum delay up to an additive margin of $5$.

\bibliographystyle{plain}

\end{document}